\documentclass[11pt,USletter]{article}
\usepackage{fullpage}

\usepackage[utf8]{inputenc}
\usepackage{amsthm,amsmath,amssymb}
\usepackage{thmtools}
\usepackage{subcaption,thm-restate}
\usepackage[mathcal,mathscr]{eucal}
\usepackage{epsfig,graphicx,graphics,color}
\usepackage{bm}
\usepackage{enumerate}
\usepackage[sort,nocompress]{cite}
\usepackage{hyperref}
\hypersetup{
    colorlinks=true,       
    linkcolor=blue,        
    citecolor=red,         
    filecolor=magenta,     
    urlcolor=cyan,         
    linktocpage=true
}
  
\theoremstyle{plain}
\newtheorem{theorem}{Theorem}
\newtheorem{lemma}[theorem]{Lemma}
\newtheorem{corollary}[theorem]{Corollary}
\newtheorem{claim}[theorem]{Claim}

\newtheorem{conjecture}{Conjecture}

\theoremstyle{definition}
\newtheorem{remark}[theorem]{Remark}

\def\final{0}  
\def\iflong{\iffalse}
\ifnum\final=0  
\usepackage[dvipsnames]{xcolor}
\newcommand{\kristof}[1]{{\color{red}[{\tiny \textbf{Kristóf:} \bf #1}]\marginpar{\color{red}*}}}
\newcommand{\yusuke}[1]{{\color{red}[{\tiny \textbf{Yusuke:} \bf #1}]\marginpar{\color{red}*}}}
\newcommand{\nao}[1]{{\color{red}[{\tiny \textbf{Nao:} \bf #1}]\marginpar{\color{red}*}}}
\else 
\newcommand{\kristof}[1]{}
\newcommand{\yusuke}[1]{}
\newcommand{\nao}[1]{}
\fi

\DeclareMathOperator*{\argmin}{arg\,min}
\DeclareMathOperator*{\argmax}{arg\,max}

\newcommand{\bR}{\mathbb{R}}
\newcommand{\cB}{\mathcal{B}}
\newcommand{\cI}{\mathcal{I}}

\newcommand{\suppp}{{\rm supp}\sp{+}}
\newcommand{\suppm}{{\rm supp}\sp{-}}

\title{Market Pricing for Matroid Rank Valuations}

\author{
Kristóf Bérczi\thanks{MTA-ELTE Egerváry Research Group, Department of Operations Research, Eötvös Loránd University, Budapest, Hungary. Email: \texttt{berkri@cs.elte.hu}.}
\and
Naonori Kakimura\thanks{Department of Mathematics, Faculty of Science and Technology, Keio University, Yokohama, Japan. Email: \texttt{kakimura@math.keio.ac.jp}.}
\and
Yusuke Kobayashi\thanks{Research Institute for Mathematical Sciences (RIMS), Kyoto University, Kyoto, Japan. Email: \texttt{yusuke@kurims.kyoto.ac.jp}.}
}

\begin{document}

\maketitle

\begin{abstract}
In this paper, we study the problem of maximizing social welfare in combinatorial markets through pricing schemes. We consider the existence of prices that are capable to achieve optimal social welfare without a central tie-breaking coordinator. In the case of two buyers with rank valuations, we give polynomial-time algorithms that always find such prices when one of the matroids is a simple partition matroid or both matroids are strongly base orderable. This result partially answers a question raised by D\"uetting and V\'egh in 2017. We further formalize a weighted variant of the conjecture of D\"uetting and V\'egh, and show that the weighted variant can be reduced to the unweighted one based on the weight-splitting theorem for weighted matroid intersection by Frank. We also show that a similar reduction technique works for M${}^\natural$-concave functions, or equivalently, gross substitutes functions.
\end{abstract}

\section{Introduction}\label{sec:intro}

In this paper, we study the problem of maximizing social welfare in combinatorial markets through pricing schemes. Let us consider a combinatorial market consisting of indivisible goods and buyers, where each buyer has a valuation function that describes the buyer's preferences over the subsets of items. The goal is to allocate the items to buyers in such a way that the social welfare, that is, the total sum of the buyers' values, is maximized. Such an allocation can be found efficiently under reasonable assumptions on the valuations \cite{nisan2006communication}. As an application of the Vickrey--Clarke--Groves (VCG) mechanism \cite{clarke1971multipart,vickrey1961counterspeculation,groves1973incentives} for welfare maximization, the VCG auction is another illustrious example. However, the problem becomes much more intricate if the optimal welfare is ought to be achieved using simpler mechanisms employed in real world markets, such as pricing.

In a pricing scheme, the seller sets the item prices, and the \emph{utility} of a buyer for a given bundle of items is defined as the value of the bundle with respect to the buyer's valuation, minus the total price of the items in the bundle. Ideally, the prices are set in such a way that there exists an allocation of the items to buyers in which the market clears and everyone receives a bundle that maximizes her utility. A pair of pricing and allocation possessing these properties is called a \emph{Walrasian equilibrium}\footnote{Walrasian equilibrium is often called competitive pricing, or market equilibrium in the literature.}, while we will refer to the price vector itself as \emph{Walrasian pricing}.  The fundamental notion of Walrasian equilibrium first appeared in \cite{walras1896elements}, and the definition immediately implies that the allocation in a Walrasian equilibrium maximizes social welfare. Therefore, the problem might seem to be settled for markets that admit such an equilibrium. 

Cohen-Addad et al. \cite{cohen2016invisible} observed that Walrasian prices alone are not sufficient to coordinate the market. The reason is that ties among different bundles have to be broken up carefully by a central coordinator, in a manner consistent with the corresponding optimal allocation. However, in real markets, buyers walk into the shop in an arbitrary sequential order and choose an arbitrary best bundle for themselves without caring about social optimum. In their paper, it is shown that the absence of a tie-braking rule may result in an arbitrary bad allocation. 

To overcome these difficulties, Cohen-Addad et al. \cite{cohen2016invisible} introduced the notion of \emph{dynamic pricing schemes}. In this setting, the seller is allowed to dynamically update the prices between buyer arrivals. Achieving optimal social welfare based on dynamic pricing would be clearly possible if the order in which buyers arrive was known in advance. Nevertheless, determining an optimal dynamic pricing scheme is highly non-trivial when the prices need to be set before getting access to the preferences of the next buyer.

The main open problem in~\cite{cohen2016invisible} asked whether any market with \emph{gross substitutes valuations} has a dynamic pricing scheme that achieves optimal social welfare. A market with gross substitutes valuations is known to be an important class of markets having Walrasian prices~\cite{kelso1982job}. It is worth noting that the existence of an optimal scheme reduces to the existence of an appropriate initial price vector; an optimal allocation then can be determined by induction. For a formal definition, we refer the reader to \cite{berger2020power}.

As a starting step towards understanding the general case, D\"utting and V\'egh~\cite{vegh} suggested to look at matroid rank functions as valuations, because a matroid rank function is a fundamental example of gross substitutes valuations. In particular, they proposed the following conjecture for the case of two buyers.\footnote{D\"utting and V\'egh conjectured that the price vector $p$ can be chosen to have all different values, that is, $p(s_1)\neq p(s_2)$ for $s_1\neq s_2$. This difference is not essential, because we can apply a perturbation to $p$ without affecting the requirements in Conjecture~\ref{conj:01}.} Here, a matroid with a ground set $S$ and a base family ${\cal B}$ is denoted by $M=(S,{\cal B})$ and we denote $p(X) := \sum_{s \in X} p(s)$ for $p:S \to \mathbb{R}$ and $X \subseteq S$.

\begin{conjecture} \label{conj:01}
Let $M_1=(S, \mathcal{B}_1)$ and $M_2=(S, \mathcal{B}_2)$ be matroids with a common ground set $S$ such that   there exist disjoint bases $B_1 \in \mathcal{B}_1$ and $B_2 \in \mathcal{B}_2$ with $B_1 \cup B_2 = S$.  Then, there exists a function $p: S \to \mathbb{R}$ (called a {\em price vector}) satisfying the following conditions.  
\begin{enumerate}
\item For any $B_1 \in \argmin_{X \in \mathcal{B}_1} p(X)$, it holds that $S \setminus B_1 \in \mathcal{B}_2$. 
\item For any $B_2 \in \argmin_{X \in \mathcal{B}_2} p(X)$, it holds that $S \setminus B_2 \in \mathcal{B}_1$. 
\end{enumerate}
\end{conjecture}

The requirements in the conjecture can be interpreted as follows. There are two buyers and each buyer $i \in \{1,2\}$ wants to buy a set of items that forms a basis in $\mathcal{B}_i$. If buyer $i$ comes to a shop first, then she chooses a cheapest set $B_i$ in $\mathcal{B}_i$ with an arbitrary tie-breaking rule. Regardless of the choice of $B_i$, the remaining set $S \setminus B_i$ is a desired set for the other buyer.  

Actually, Conjecture~\ref{conj:01} resolves the existence of a static pricing scheme for a two-buyer market with matroid rank valuations. That is, if Conjecture~\ref{conj:01} is true, then the following conjecture is also true. See Lemma~\ref{lem:00to01} for the details.  

\begin{conjecture} \label{conj:00}
Let $M_1=(S, \mathcal{B}_1)$ and $M_2=(S, \mathcal{B}_2)$ be matroids with rank functions $r_1$ and $r_2$, respectively.  Then, there exists a function $p: S \to \mathbb{R}$ satisfying the following conditions.  
\begin{enumerate}
\item For any $B_1 \in \argmax_{X\subseteq S} (r_1(X)-p(X))$ and for any $B_2 \in \argmax_{Y\subseteq S \setminus B_1} (r_2(Y)-p(Y))$,  $r_1(B_1)+r_2(B_2) = \max \{r_1(X) + r_2(Y) \mid X, Y \subseteq S,\ X\cap Y=\emptyset \}$. 
\item
For any $B_2 \in \argmax_{Y\subseteq S} (r_2(Y)-p(Y))$ and for any $B_1 \in \argmax_{X\subseteq S \setminus B_2} (r_1(X)-p(X))$,  $r_1(B_1)+r_2(B_2) = \max \{r_1(X) + r_2(Y) \mid X, Y \subseteq S,\ X\cap Y=\emptyset \}$. 
\end{enumerate}
\end{conjecture}


In the conjecture, if buyer $i$ comes to a shop first, then she chooses an arbitrary bundle $B_i$ that maximizes her utility $r_i-p$, and the second buyer chooses a best bundle in $S \setminus B_i$.
The requirements mean that any choice of $B_i$ results in an allocation maximizing the social welfare. 
Thus, whoever comes first, we can achieve the optimal social welfare.

\paragraph*{Previous work}

The notion of Walrasian equilibrium dates back to 1874 \cite{walras1896elements}, originally defined for divisible goods. In their analysis of the matching problem, Kelso and Crawford \cite{kelso1982job} introduced the so-called gross substitutes condition, and showed the existence of Walrasian prices for gross substitutes valuations.
Gul and Stacchetti \cite{gul1999walrasian} later verified that, in a sense, this condition is necessary to ensure the existence of a Walrasian equilibrium.\footnote{The simplest example of gross substitutes valuations are unit demand preferences, when each agent can enjoy at most one item. Gul and Stacchetti showed that gross substitutes preferences form the largest set containing unit demand preferences for which an existence theorem can be obtained.} 

It was first observed by Cohen-Addad et al.~\cite{cohen2016invisible} and Hsu et al.~\cite{hsu2016prices} that Walrasian prices are not sufficient to control the market as ties must be broken in a coordinated fashion that is consistent with maximizing social welfare. A natural idea for resolving this issue would be trying to find Walrasian prices where ties do not occur. However, Hsu et al.\ showed that  minimal Walrasian prices always induce ties. Even more, Cohen-Addad et al.\ proved that no static prices can give more than $2/3$ of the social welfare when buyers arrive sequentially. As a workaround, they proposed a dynamic pricing scheme for matching markets (i.e., unit-demand valuations), where the prices are updated between buyer-arrivals based upon the current inventory without knowing the identity of the next buyer. On the negative side, they presented a market with coverage valuations where Walrasian prices do exist, but no dynamic pricing scheme can achieve the optimal social welfare. Meanwhile, Hsu et al.\ showed that, under certain conditions, minimal Walrasian equilibrium prices induce low over-demand and high welfare. Recently, Berger et al.~\cite{berger2020power} considered markets beyond unit-demand valuations, and gave a characterization of all optimal allocations in multi-demand markets. Based on this, they provided a polynomial-time algorithm for finding optimal dynamic prices up to three multi-demand buyers.

To overcome the limitations of Walrasian equilibrium, Feldman et al.~\cite{feldman2016combinatorial} proposed a relaxation called combinatorial Walrasian equilibrium in which the seller can partition the items into indivisible bundles prior to sale, and they provided an algorithm that determines bundle prices obtaining at least half of the optimal social welfare.

Another line of research concentrated on posted-price mechanisms in online settings. As alternatives to optimal auctions, Blumrosen and Holenstein \cite{blumrosen2008posted} studied posted-price mechanisms and dynamic auctions in Bayesian settings under the objective of maximizing revenue. They gave a characterization of the optimal revenue for general distributions, and provided algorithms that achieve the optimal solution. Chawla et al.~\cite{chawla2010multi,chawla2010power} developed a theory of sequential posted-price mechanisms, and provided constant-factor approximation algorithms for several multi-dimensional multi-unit auction problems and generalizations to matroid feasibility constraints. In \cite{feldman2014combinatorial}, Feldman et el.\ verified the existence of prices that, in expectation, achieve at least half of the optimal social welfare for fractionally subadditive valuations, a class that includes all submodular functions. D\"utting et al.~\cite{dutting2016posted,dutting2017prophet} provided a general framework for posted-price mechanisms in Bayesian settings. Chawla et al.~\cite{chawla2019pricing} showed that static, anonymous bundle pricing mechanisms are useful when buyers' preferences have complementarities. Ezra et al.~\cite{ezra2018pricing} provided upper and lower bounds on the largest fraction of the optimal social welfare that can be guaranteed with static prices for several classes of valuations, such as submodular, XOS, or subadditive. A setting related to online bipartite matching, called the Max-Min Greedy matching, was considered in \cite{eden2019max}.

\paragraph*{Our results}

In the present paper, we concentrate on combinatorial markets with two buyers having matroid rank valuations, where the matroid corresponding to buyer $i$ is denoted by $M_i=(S,\cB_i)$ for $i=1,2$. 
Since this setting is reduced to Conjecture~\ref{conj:01}, in which each buyer has to buy a set of items that forms a basis of a matroid, we focus on Conjecture~\ref{conj:01}. 

While Conjecture~\ref{conj:01} remains open in general, we give polynomial-time\footnote{In matroid algorithms, it is usually assumed that the matroids are accessed through independence oracles, and the complexity of an algorithm is measured by the number of oracle calls and other conventional elementary steps.} algorithms for two important special cases. In the first one, one of the matroids is a partition matorid. Although partition matroids have relatively simple structure, finding the proper price vector $p$ is non-trivial even in this seemingly simple case.

\begin{restatable}{theorem}{partition}
\label{thm:partition}
Let $M_1$ be a partition matroid with partition classes of size at most $2$ and with all-ones upper bound on the partition classes, and let $M_2$ be an arbitrary matroid. Then Conjectures~\ref{conj:01} and~\ref{conj:00} hold, and a price vector $p$ satisfying the conditions can be computed in polynomial time.
\end{restatable}

Next we consider strongly base orderable matroids, a class of matroids with distinctive structural properties. 
Roughly, in a strongly base orderable matroid, for any pair of bases, there exists a bijection between them satisfying a certain property (see Section~\ref{sec:prelim} for the formal definition). 
Note that various matroids appearing in combinatorial and graph optimization problems belong to this class, such as partition, laminar, transversal matroids, or more generally, gammoids. 

\begin{restatable}{theorem}{sbo}
\label{thm:sbo}
If both $M_1$ and $M_2$ are strongly base orderable, then Conjectures~\ref{conj:01} and~\ref{conj:00} hold. Furthermore, a price vector $p$ satisfying the conditions can be computed in polynomial time if, for any pair of bases, the bijection between them can be computed in polynomial time.
\end{restatable}

Another contribution of this paper is to show the equivalence between Conjecture~\ref{conj:01} and its weighted counterpart as below.

\begin{conjecture} \label{conj:03}
For $i \in \{1, 2\}$, let $M_i=(S, \mathcal{B}_i)$ be a matroid and $w_i: S \to \mathbb{R}$ be a weight function.   Assume that there exist disjoint bases $B_1 \in \mathcal{B}_1$ and $B_2 \in \mathcal{B}_2$ with $B_1 \cup B_2 = S$.  Then, there exists a function $p: S \to \mathbb{R}$ satisfying the following conditions.  
\begin{enumerate}
\item For any $B_1 \in \argmax_{X \in \mathcal{B}_1} (w_1(X) - p(X))$, we have that 
$B_1$ is a maximizer of $w_1(X) + w_2(S \setminus X)$ subject to $X \in \mathcal{B}_1$ and $S \setminus X \in \mathcal{B}_2$.  
\item For any $B_2 \in \argmax_{X \in \mathcal{B}_2} (w_2(X) - p(X))$, we have that 
$B_2$ is a maximizer of $w_1(S \setminus X) + w_2(X)$ subject to $S \setminus X \in \mathcal{B}_1$ and $X \in \mathcal{B}_2$.  
\end{enumerate}
\end{conjecture}

Clearly, Conjecture~\ref{conj:01} is a special case of Conjecture~\ref{conj:03}; this follows easily by setting $w_1 \equiv w_2 \equiv 0$. Somewhat surprisingly, the reverse implication also holds for arbitrary matroids. 

\begin{restatable}{theorem}{unweighttoweight}
\label{thm:unweight2weight}
If Conjecture~\ref{conj:01} is true, then Conjecture~\ref{conj:03} is also true. 
\end{restatable}

More generally, we prove that Theorem~\ref{thm:unweight2weight} can be generalized to the case with gross substitutes valuations, i.e., M${}^\natural$-concave functions.
See Theorem~\ref{thm:MnaturalReduction} in Section~\ref{sec:Mnatural} for the details.

Based on Theorem~\ref{thm:unweight2weight} and the properties of partition and strongly base orderable matroids, we have the following corollaries.

\begin{corollary}
\label{cor:partitionw}
Let $M_1$ be a partition matroid with partition classes of size at most $2$ and with all-ones upper bound on the partition classes, and let $M_2$ be an arbitrary matroid. 
Then Conjecture~\ref{conj:03} holds, and a price vector $p$ satisfying the conditions can be computed in polynomial time.
\end{corollary}

\begin{corollary}
\label{cor:sbow}
If both $M_1$ and $M_2$ are strongly base orderable, then Conjecture~\ref{conj:03} holds. Furthermore, a price vector $p$ satisfying the conditions can be computed in polynomial time if, for any pair of bases, the bijection between them can be computed in polynomial time. 
\end{corollary}


\paragraph*{Paper organization}

The rest of the paper is organized as follows. Basic definitions and notation are given in Section~\ref{sec:prelim}. Theorems~\ref{thm:partition} and \ref{thm:sbo} are proved in Sections~\ref{sec:part} and \ref{sec:sbo}, respectively. The connection between unweighted and weighted variants of the problem is discussed in Section~\ref{sec:weighted}. The reduction technique is extended to gross substitutes valuations in Section~\ref{sec:Mnatural}. We conclude the paper in Section~\ref{sec:conc}.

\section{Preliminaries}\label{sec:prelim}

\paragraph*{Basic notation}

The sets of reals, non-negative reals, integers, and non-negative integers are denoted by $\mathbb{R}$, $\mathbb{R}_+$, $\mathbb{Z}$, and $\mathbb{Z}_+$, respectively. Let $S$ be a finite set. Given a subset $B\subseteq S$ and elements $x,y\in S$, we write $B-x+y$ for short to denote the set $(B\setminus\{x\})\cup\{y\}$. The \emph{symmetric difference} of two sets $X$ and $Y$ is $X\triangle Y:=(X\setminus Y)\cup(Y\setminus X)$. For a function $f:S\to\mathbb{R}$, we use $f(X):=\sum_{x\in X} f(x)$. For two vectors $x, y \in \mathbb{R}^S$, we denote $x \cdot y := \sum_{s\in S} x(s) y(s)$.

\paragraph*{Matroids and matroid intersection}

Matroids were introduced as an abstract generalization of linear independence in vector spaces \cite{whitney1992abstract,nishimura2009lost}. A \emph{matroid} $M$ is a pair $(S,\cI)$ where $S$ is the \emph{ground set} of the matroid and $\cI\subseteq 2^S$ is the family of \emph{independent sets} satisfying the \emph{independence axioms}: (I1) $\emptyset\in\cI$, (I2) $X\subseteq Y\in \cI\Rightarrow X\in\cI$, and (I3) $X,Y\in\cI,\ |X|<|Y|\Rightarrow\exists e\in Y\setminus X\ \text{s.t.}\ X+e\in\cI$. A \emph{loop} is an element that is non-independent on its own. The \emph{rank} of a set $X\subseteq S$ is the maximum size of an independent set contained in $X$, and is denoted by $r(X)$. Here $r$ is called the \emph{rank function} of $M$. Maximal independent sets of $M$ are called \emph{bases} and their set is denoted by $\cB$. Alternatively, matroids can be defined through the \emph{basis axioms}: (B1) $\cB\neq\emptyset$, and (B2) $B_1,B_2\in\cB, x\in B_1\setminus B_2\Rightarrow\exists y\in B_2\setminus B_1\ \text{s.t.}\ B_1-x+y\in\cB$. 
In this paper, a matroid is denoted by a pair $(S, {\cal B})$, where $S$ is a ground set and ${\cal B}$ is a base family.

For a matroid $M=(S, {\cal B})$ and for $T \subseteq S$, \emph{deleting} $T$ gives a matroid $M'$ on the ground set $S \setminus T$ such that a subset of $S \setminus T$ is independent in $M'$ if and only if it is independent in $M$. For $T \subseteq S$, \emph{contracting} $T$ gives a matroid $M'$ on the ground set $S \setminus T$ whose rank function is $r'(X) = r(X \cup T) - r(T)$, where $r$ is the rank function of $M$. \emph{Adding a parallel copy} of an element $s\in S$ gives a new matroid $M'=(S',\cB')$ on ground set $S'=S+s'$ where $\cB'=\{X\subseteq S':\ \text{either}\ X\in\cB,\ \text{or}\ s\notin X,\ s'\in X\ \text{and}\ X-s'+s\in\cB\}$.
The \emph{direct sum} $M_1\oplus M_2$ of matroids $M_1=(S_1,\cB_1)$ and $M_2=(S_2,\cB_2)$ on disjoint ground sets is a matroid $M=(S_1\cup S_2,\cB)$ whose bases are the disjoint unions of a basis $M_1$ and a basis of $M_2$.
The \emph{sum} or \emph{union} $M_1+M_2$ of $M_1=(S,\cB_1)$ and $M_2=(S,\cB_2)$ on the same ground set is a matroid $M=(S,\cB)$ whose independent sets are the disjoint unions of an independent set of $M_1$ and an independent set of $M_2$. 


For a basis $B\in\cB$, let us consider the bipartite graph $G=(S,E[B])$ defined by $E[B]:=\{(x,y)\mid x\in B,\, y\in S\setminus B,\, B-x+y\in\cB\}$. Krogdahl \cite{krogdahl1974combinatorial,krogdahl1976combinatorial,krogdahl1977dependence} verified the following statement (see also \cite[Theorem 39.13]{schrijver2003combinatorial}).

\begin{theorem}[Krogdahl]\label{thm:krog}
Let $M=(S,\cB)$ be a matroid and let $B\in\cB$. Let $B'\subseteq S$ be such that $|B|=|B'|$ and $E[B]$ contains a unique perfect matching on $B\triangle B'$. Then $B'\in\cB$.
\end{theorem}

In the \emph{weighted matroid intersection problem}, we are given two matroids $M_1=(S,\cB_1)$ and $M_2=(S,\cB_2)$ on the same ground set together with a weight function $w:S\to\bR$, and the goal is to find a common basis maximizing $w(B)$, that is, $B\in\argmax\{w(B)\mid B\in\cB_1\cap\cB_2\}$. The celebrated weight-splitting theorem of Frank~\cite{frank1981weighted} gives a min-max relation for the weighted matroid intersection.

\begin{theorem}[Frank]\label{thm:Frank}
The maximum $w$-weight of a common basis of $M_1=(S,\cB_1)$ and $M_2=(S,\cB_2)$ is equal to the minimum of $\max\{w_1(B) \mid B \in\cB_1\} + \max\{w_2(B) \mid B \in\cB_2\}$ subject to $w=w_1+w_2$.
In particular, for an optimal weight-splitting $w=w_1+w_2$, 
it holds that $\argmax\{w(B)\mid B\in\cB_1\cap\cB_2\} = \argmax\{w_1(B)\mid B\in\cB_1\} \cap \argmax\{w_2(B)\mid B\in\cB_2\}$.
\end{theorem}

A \emph{$k$-uniform matroid} is a matroid $M=(S,\cB)$ where $\cB=\{X\subseteq S\mid |X|=k\}$ for some $k\in\mathbb{Z}_+$. A \emph{partition matroid} $M=(S,\cB)$ is the direct sum of uniform matroids, or in other words, $\cB=\{X\subseteq S\mid |X\cap S_i|=k_i\ \text{for}\ i=1,\dots,q\}$ for some partition $S=S_1\cup\dots\cup S_q$ of $S$ and $k_i\in\mathbb{Z}_+$ for $i=1,\dots,q$. Each $S_i$ is called a \emph{partition class}. In the paper, we will work with partition matroids satisfying $|S_i|\le 2$ and $k_i=1$ for $i=1,\dots,q$.

For further details on matroids and the matroid intersection problem, we refer the reader to \cite{oxley2011matroid, schrijver2003combinatorial}.

\paragraph*{Dual matroids}

The \emph{dual} of a matroid $M=(S,\cB)$ is the matroid $M^*=(S,\cB^*)$ where $\cB^*=\{B^*\subseteq S\mid S\setminus B^*\in\cB\}$. Given one of the standard oracles for $M$, the same oracle can be constructed for $M^*$ as well.   

We now rephrase Conjecture~\ref{conj:01} by using dual matroids. Suppose that $M_1$ and $M_2$ are matroids as in Conjecture~\ref{conj:01} and let $M^*_2 = (S, \mathcal{B}^*_2)$ be the dual matroid of $M_2$. Then, we can see that $S \setminus B_1 \in \mathcal{B}_2$ is equivalent to $B_1 \in \mathcal{B}^*_2$, and  $B_2 \in \argmin_{X \in \mathcal{B}_2} p(X)$ is equivalent to $S \setminus B_2 \in \argmax_{X \in \mathcal{B}^*_2} p(X)$. Therefore, by replacing $M_2$ and $S \setminus B_2$ with $M^*_2$ and $B_2$, respectively, Conjecture~\ref{conj:01} is equivalent to the following conjecture. 

\begin{conjecture} \label{conj:02}
Let $M_1=(S, \mathcal{B}_1)$ and $M_2=(S, \mathcal{B}_2)$ be matroids with a common ground set $S$ such that there exists a common basis $B \in \mathcal{B}_1 \cap \mathcal{B}_2$.  Then, there exists a function $p: S \to \mathbb{R}$ satisfying the following conditions.  
\begin{enumerate}
\item For any $B_1 \in \argmin_{X \in \mathcal{B}_1} p(X)$, it holds that $B_1 \in \mathcal{B}_2$. 
\item For any $B_2 \in \argmax_{X \in \mathcal{B}_2} p(X)$, it holds that $B_2 \in \mathcal{B}_1$. 
\end{enumerate}
\end{conjecture}

Conjecture~\ref{conj:02} bears a lot of similarities with the problem of packing common bases in the intersection of two matroids. If $M_1$ and $M_2$ share two disjoint common bases then setting the prices low on one of them and high on the other gives a desired $p$. If $S$ can be partitioned into two disjoint bases in both $M_1$ and $M_2$, then the statement may be reminiscent of Rota's famous conjecture concerning rearrangements of bases \cite{huang1994relations}.

\paragraph*{Strongly base orderable matroids}

A matroid $M = (S, \mathcal{B})$ is \emph{strongly base orderable} if  for any two bases $B_1, B_2 \in \mathcal{B}$, there exists a bijection $f: B_1 \to B_2$ such that  $(B_1 \setminus X) \cup f(X) \in \mathcal{B}$ for any $X \subseteq B_1$,  where we denote $f(X) := \{ f(e) \mid e \in X\}$. Davies and McDiarmid~\cite{davies1976disjoint} observed the following~(see also \cite[Theorem 42.13]{schrijver2003combinatorial}). 

\begin{theorem}[Davies and McDiarmid]
\label{thm:DM76}
Let $M_1=(S,\cB_1)$ and $M_2=(S,\cB_2)$ be strongly base orderable matroids. If $X\subseteq S$ can be partitioned into $k$ bases in both $M_1$ and $M_2$, then $X$ can be partitioned into $k$ common bases.
Furthermore, such $k$ common bases can be computed in polynomial time if the bijection $f$ can be computed in polynomial time for any pair of bases.
\end{theorem}

The following technical lemma about strongly base orderable matroids will be used in the proof of Corollary~\ref{cor:sbow}. 

\begin{lemma}
\label{lem:maximizerSBO}
Let $M = (S, \mathcal{B})$ be a strongly base orderable matroid, $q: S \to \mathbb{R}$ be a function, and define a matroid $\hat{M} = (S, \hat{\mathcal{B}})$ by $\hat{\mathcal{B}} = \argmax_{X \in \mathcal{B}} q(X)$. Then $\hat{M}$ is strongly base orderable. 
\end{lemma}

\begin{proof}
Let $B_1, B_2 \in \hat{\mathcal{B}}$. Since both $B_1$ and $B_2$ are bases of $M = (S, \mathcal{B})$, there exists a bijection $f: B_1 \to B_2$ such that $(B_1 \setminus X) \cup f(X) \in \mathcal{B}$ for any $X \subseteq B_1$.  Since $q(B_1) \ge q( (B_1 \setminus X) \cup f(X))$ for any $X \subseteq B_1$ by $B_1 \in \hat{\mathcal{B}}$, it holds that $q(X) \ge q(f(X))$.  In particular, $q(x) \ge q(f(x))$ for any $x \in B_1$.  Since $B_2 \in \hat{\mathcal{B}}$, we obtain $q(B_1) = q(B_2) = q(f(B_1))$, which shows that $q(x) = q(f(x))$ for any $x \in B_1$.  Therefore, $q(B_1) = q( (B_1 \setminus X) \cup f(X))$ for any $X \subseteq B_1$, and hence  $(B_1 \setminus X) \cup f(X) \in \hat{\mathcal{B}}$. This shows that $\hat{M}$ is strongly base orderable. 
\end{proof}

\paragraph*{Market model}

In a combinatorial market, we are given a set $S$ of \emph{indivisible items} and a set $J$ of  \emph{buyers}. Each buyer $i\in J$ has a \emph{valuation function} $v_i:2^S\to\bR$ that describes the buyer's preferences over the subsets of items. 
Given \emph{prices} $p:S\to\bR$, the \emph{utility} of buyer $i\in J$ for a subset $X\subseteq S$ is defined by $u_i(X)=v_i(X)-p(X)$. 
The buyers arrive in an undetermined order, and the next buyer always picks a subset of items that maximizes her utility. The goal is to set the prices in such a way that no matter which buyer arrives next, the final allocation of items maximizes the social welfare. In a dynamic pricing scheme, the prices can be updated between buyer arrivals based on the remaining sets of items and buyers.

We focus on the case of two buyers with matroid rank functions as valuations. Let $M_1=(S,\cB_1)$ and $M_2=(S,\cB_2)$ be matroids with rank functions $r_1$ and $r_2$, respectively. The valuation of agent $i$ is $r_i$ for $i=1,2$. The valuations are accessed through one of the standard matroid oracles (e.g. independence or rank oracle). As described in the introduction, this setting can be reduced to the case in which each buyer always chooses a basis that maximizes her utility, that is, Conjecture~\ref{conj:00} can be reduced to Conjecture~\ref{conj:01}.


\begin{lemma}
\label{lem:00to01}
If Conjecture~\ref{conj:01} is true, then Conjecture~\ref{conj:00} is also true. 
\end{lemma}

\begin{proof}
Let $M_1=(S, \mathcal{B}_1)$ and $M_2=(S, \mathcal{B}_2)$ be matroids as in Conjecture~\ref{conj:00} and let $\hat B_1 \in \mathcal{B}_1$ and $\hat B_2 \in \mathcal{B}_2$ be a pair of bases that maximizes $|\hat B_1 \cup \hat B_2|$. 
For $i \in \{1, 2\}$, let $M'_i$ be the matroid obtained from $M_i$ by deleting $S \setminus (\hat B_1 \cup \hat B_2)$ and contracting $\hat B_1 \cap \hat B_2$. Then, $M'_1=(S', \mathcal{B}'_1)$ and $M'_2=(S', \mathcal{B}'_2)$ are matroids with a common ground set $S':=(\hat B_1 \cup \hat B_2) \setminus (\hat B_1 \cap \hat B_2)$ such that there exist disjoint bases $\hat B_1 \setminus \hat B_2 \in \mathcal{B}'_1$ and $\hat B_2 \setminus \hat B_1 \in \mathcal{B}'_2$ whose union is $S'$. 
Hence, by assuming that Conjecture~\ref{conj:01} is true, there exists a price vector $p': S' \to \mathbb{R}$ with the following conditions.  
\begin{enumerate}
\item For any $B'_1 \in \argmin_{X \in \mathcal{B}'_1} p'(X)$, it holds that $S' \setminus B'_1 \in \mathcal{B}'_2$. 
\item For any $B'_2 \in \argmin_{X \in \mathcal{B}'_2} p'(X)$, it holds that $S' \setminus B'_2 \in \mathcal{B}'_1$. 
\end{enumerate}
We observe that we can modify the price vector $p'$ so that $0 < p'(s) < 1$ for every $s \in S'$, by replacing $p'(s)$ with $\alpha \cdot p'(s) + \beta$ for some $\alpha>0$ and $\beta \in \mathbb{R}$. By using such a function $p'$, define $p: S \to \mathbb{R}$ by 
$$
p(s) = 
\begin{cases}
p'(s) & \text{if $s \in S'$,} \\
0 & \text{if $s \in \hat B_1 \cap \hat B_2$,} \\
1 & \text{if $s \in S \setminus (\hat B_1 \cup \hat B_2)$.} 
\end{cases}
$$
For $B_1 \in \argmax_{X \subseteq S} (r_1(X)-p(X))$, the definition of $p$ shows that $B_1 = B'_1 \cup (\hat B_1 \cap \hat B_2)$ for some $B'_1 \in \argmin_{X \in \mathcal{B}'_1} p'(X)$. Since this implies $S' \setminus B'_1 \in \mathcal{B}'_2$, it holds that $S' \setminus B'_1$ is a maximal independent set of $M_1$ in $S \setminus B_1$ by the maximality of $|\hat B_1 \cup \hat B_2|$. Therefore, if $B_2 \in \argmax_{Y\subseteq S \setminus B_1} (r_2(Y)-p(Y))$, then $B_2 = S' \setminus B'_1$ and hence 
\begin{align*}
r_1(B_1)+r_2(B_2) 
&= |B'_1| +|\hat B_1 \cap \hat B_2| + |S' \setminus B'_1| = |\hat B_1 \cup \hat B_2| \\
&= \max \{r_1(X) + r_2(Y) \mid X, Y \subseteq S,\ X\cap Y=\emptyset \},  
\end{align*}
which shows the first requirement of Conjecture~\ref{conj:00}. The same argument works for $B_2 \in \argmax_{X \subseteq S} (r_2(X)-p(X))$. Therefore, $p$ satisfies the requirements in Conjecture~\ref{conj:00}. 
\end{proof}

Note that a pair of bases $\hat B_1 \in \mathcal{B}_1$ and $\hat B_2 \in \mathcal{B}_2$ maximizing $|\hat B_1 \cup \hat B_2|$ can be computed in polynomial time by applying a matroid intersection algorithm to $M_1$ and $M^*_2$. Note also that the price vector $p$ obtained in the above proof is not necessarily a Walrasian price.

We can consider a weighted variant of Conjecture~\ref{conj:01} in which we are given weight functions $w_1:S\to\bR$ and $w_2:S\to\bR$. For a buyer $i\in \{1,2\}$ and for a basis $X\in \mathcal{B}_i$, the valuation $v_i(X)$ is defined as $w_i(X)$. Each buyer chooses a basis that maximizes her utility. Note that choosing a basis is a hard constraint, and hence we do not have to define $v_i(X)$ for $X \not\in \mathcal{B}_i$. The goal is to find a price vector $p$ that achieves the optimal social welfare $\max\{w_1(X)+w_2(S\setminus X)\mid X\in\cB_1,\, S\setminus X\in\cB_2\}$.

Recently, Berger et al.~\cite{berger2020power} investigated the existence of optimal dynamic pricing schemes for $k$-demand valuations. 
A valuation $v:2^S\to\bR_+$ is \emph{$k$-demand} if $v(X)=\max\{\sum_{s\in Z} v(s)\mid Z\subseteq X,\, |Z|\leq k\}$. Although this setting is similar to our weighted variant for $k$-uniform matroids, our results do not directly generalize their work because of our assumption on the buyers' choices.




\section{Partition matroids}\label{sec:part}

The aim of this section is to prove the existence of a required price vector $p$ for instances where $M_1$ is a partition matroid of special type. 


\partition*


\begin{proof}
Since Conjectures~\ref{conj:01} and \ref{conj:02} are equivalent by replacing $M_2$ with its dual $M^*_2$, we show Conjecture~\ref{conj:02}. 
Let $M_1=(S,\cB_1)$ be a partition matroid defined by partition $S=S_1\cup\dots\cup S_q$ where $|S_i|\le 2$ for $i=1,\dots,q$, that is, $\cB_1=\{X\subseteq S\mid |X\cap S_i|=1\ \text{for}\ i=1,\dots,q\}$. Let $M_2=(S,\cB_2)$ be an arbitrary matroid such that $M_1$ and $M_2$ have a common basis.  

Let $B_1\in\cB_1\cap\cB_2$ be an arbitrary common basis. Take another common basis $B_2\in\cB_1\cap\cB_2$ (possibly $B_2=B_1$) such that $|B_1 \cap B_2|$ is minimized. We consider a bipartite digraph $D=(V, E)$ defined by
\begin{align}
V &= (B_1 \cap B_2) \cup (S \setminus (B_1 \cup B_2)), \notag \\
E &= \{(x, y) \mid x \in B_1 \cap B_2,\, y \in S \setminus (B_1 \cup B_2),\, B_1 - x + y \in \mathcal{B}_1 \} \label{eq:D}\\
   & \qquad \cup \{(y, x) \mid x \in B_1 \cap B_2,\, y \in S \setminus (B_1 \cup B_2),\, B_2 - x + y \in \mathcal{B}_2 \}. \notag
\end{align}

\begin{claim}\label{clm:02b}
The digraph $D$ is acyclic.
\end{claim}
\begin{proof}
Let $x\in B_1\cap B_2$ and $y\in S\setminus (B_1\cup B_2)$. As $M_1$ is defined on a partition consisting of classes of size at most $2$, $B-x+y\in\cB_1$ implies that $\{x,y\}$ is one of the partition classes. This implies that $B_1-x+y\in\cB_1$ if and only if $B_2+x-y\in\cB_1$.

Now suppose to the contrary that $D$ contains a dicycle. Choose a dicycle $C$ with the smallest number of vertices, which implies that $C$ has no chord. Then, $B'_2:=B_2 \triangle V(C)$ is a common basis of $M_1$ and $M_2$ by the above observation and Theorem~\ref{thm:krog}. Since $|B_1\cap B'_2|<|B_1\cap B_2|$, this contradicts that $|B_1 \cap B_2|$ is minimized.
\end{proof}

Let $n=|S|$. We now consider a function $p: S \to \mathbb{R}$ satisfying the following: $p(x) := 0$ for $x \in B_1 \setminus B_2$, $p(x) := n+1$ for $x \in B_2 \setminus B_1$, $p(x)$ are distinct values in $\{1, 2, \dots , n\}$ for $x \in V$, and  $p(x) < p(y)$ for $(x, y) \in E$. Note that such a function exists by Claim~\ref{clm:02b}, which can be found easily by the topological sorting. In what follows, we show that $p$ satisfies that $\argmin_{X \in \mathcal{B}_1} p(X) = \{B_1\}$ and $\argmax_{X \in \mathcal{B}_2} p(X) = \{B_2\}$. 
\begin{claim}\label{clm:03b}
$\argmin_{X \in \mathcal{B}_1} p(X) = \{B_1\}$ and $\argmax_{X \in \mathcal{B}_2} p(X) = \{B_2\}$. 
\end{claim}
\begin{proof}
For a non-negative integer $k$, let $S_k := \{x \in S \mid p(x) \le k\}$ and let $I_k$ be a minimizer of $p(X)$ subject to $X$ being a maximal independent set of $M_1$ and $X \subseteq S_k$. Note that $I_k$ can be computed by a greedy algorithm. Since $S_n$ contains a basis $B_1$, the greedy algorithm chooses no element in $B_2 \setminus B_1$, which means that $I_k \cap (B_2 \setminus B_1) = \emptyset$ for every $k$. We also note that $I_k$ is uniquely determined for each $k$, since $p(x)$'s are distinct for $x \in V$. 

We show that $I_k = B_1 \cap S_k$ for every $k$ by induction on $k$. Since $I_0 = B_1 \setminus B_2$, it is obvious that $I_0 = B_1 \cap S_0$. Fix $k \ge 1$ and assume that $I_{k-1} = B_1 \cap S_{k-1}$. Then, we have the following. 
\begin{itemize}
\item If there exists $x \in B_1 \cap B_2$ with $p(x)=k$, then 
$I_{k} = I_{k-1} +x$, and hence $I_k = B_1 \cap S_k$. 
\item Suppose that there exists $y \in S \setminus (B_1 \cup B_2)$ with $p(y)=k$.  We show that $I_{k-1} +y$ is not independent in $M_1$. Suppose to the contrary that $I_{k-1} +y$ is independent. Then, there exists $x \in B_1 \setminus I_{k-1}$ such that $B_1 - x + y \in \mathcal{B}_1$, and hence $(x, y) \in E$. By the choice of $p$, we obtain $p(x) < p(y)$, i.e., $x \in S_{k-1}$. This contradicts $x \in B_1 \setminus I_{k-1}$, because $S_{k-1} \cap (B_1 \setminus I_{k-1}) = \emptyset$ by the induction hypothesis. Therefore, $I_{k-1} +y$ is not independent in $M_1$, which shows that $I_{k} = I_{k-1}$ and $I_k = B_1 \cap S_k$. 
\item If there exists no $x \in V$ with $p(x)=k$, then $I_{k} = I_{k-1}$, and hence $I_k = B_1 \cap S_k$. 
\end{itemize}
Therefore, $I_{k} = B_1 \cap S_{k}$ holds for every $k$ by induction. This shows that $I_{n+1} = B_1 \cap S_{n+1} = B_1$, and hence $\argmin_{X \in \mathcal{B}_1} p(X) = \{ I_{n+1} \} = \{B_1\}$.

By a similar argument, we obtain $\argmax_{X \in \mathcal{B}_2} p(X) = \{B_2\}$. 
\end{proof}

Since $B_1, B_2 \in \mathcal{B}_1 \cap \mathcal{B}_2$, this claim shows that $p$ satisfies the requirements in Conjecture~\ref{conj:02}. Thus, Conjecture~\ref{conj:02} holds, and hence Conjecture~\ref{conj:01} also holds.


This together with Lemma~\ref{lem:00to01} shows that Conjecture~\ref{conj:00} also holds. Note that, in the proof of Lemma~\ref{lem:00to01}, we modify given matroids by deleting and contracting some elements, but this modification does not affect the assumption on $M_1$. That is, if $M_1$ is a partition matroid with partition classes of size at most $2$ and with all-ones upper bound on the partition classes, then the obtained matroid $M'_1$ is also a partition matroid of this type.
\end{proof}

\begin{remark}
Note that in the proof of Theorem~\ref{thm:partition}, we fixed the basis $B_1\in\cB_1$ arbitrarily. That is, for any $B_1\in\cB_1$, the optimal price vector $p$ can be set in such a way that the maximum utility of the buyer corresponding to $M_1$ is attained on $B_1$. It is not difficult to see that the analogous statement holds for any basis $B_2\in\cB_2$.   
\end{remark}

\begin{remark}
Even when $\mathcal{B}_1$ is a base family of a partition matroid as in Theorem~\ref{thm:partition}, if $\mathcal{B}_2$ is an arbitrary set family of $S$, then the requirements in Conjecture~\ref{conj:01} do not necessarily hold. To see this, suppose that $S=\{1,2,3,4\}$, $\mathcal{B}_1=\{\{1,3\}, \{1,4\}, \{2,3\}, \{2,4\}\}$, and $\mathcal{B}_2=\{\{2,4\}, \{1,2\}, \{3,4\}\}$. Then, $(B_1, B_2) = (\{1,3\}, \{2,4\})$ is a unique pair of disjoint sets such that $B_1 \in \mathcal{B}_1$, $B_2 \in \mathcal{B}_2$, and $B_1 \cup B_2 =S$. If $p$ satisfies the requirements in Conjecture~\ref{conj:01}, then $p(1)<p(2)$ and $p(3)<p(4)$ hold by the first requirement and $p(4)<p(1)$ and $p(2)<p(3)$ hold by the second requirement. This shows that such $p$ does not exist. 
\end{remark}

\section{Strongly base orderable matroids}\label{sec:sbo}

In this section, we show that Conjectures~\ref{conj:01} and~\ref{conj:00} hold for strongly base orderable matroids. 
The proof is based on a similar approach to that of Theorem~\ref{thm:partition}. Nevertheless, there are small but crucial differences.

\sbo*


\begin{proof}
In order to show Conjecture~\ref{conj:01}, we first show Conjecture~\ref{conj:02} under the assumption that $M_1$ and $M_2$ are strongly base orderable. Let $M_1=(S, \mathcal{B}_1)$ and $M_2=(S, \mathcal{B}_2)$ be strongly base orderable matroids that have a common basis.  We take two common bases $B_1, B_2 \in \mathcal{B}_1 \cap \mathcal{B}_2$ (possibly $B_1 = B_2$) such that $|B_1 \cap B_2|$ is minimized.  For each element $x \in S$, we add a parallel copy $x'$ of $x$ to the matroid $M_i$ and denote the matroid thus obtained by $M^+_i = (S \cup S', \mathcal{B}^+_i)$ for $i\in\{1,2\}$.  We denote $X' := \{x' \mid x \in X\}$ for $X \subseteq S$.  Let $2M^+_i = (S \cup S', 2 \mathcal{B}^+_i)$ be the sum of two copies of $M^+_i$. As $M^+_i$ clearly has two disjoint bases, we have $2 \mathcal{B}^+_i := \{ Y_1 \cup Y_2 \mid Y_1 \text{ and } Y_2 \text{ are disjoint bases of } M^+_i \}. $

\begin{claim}\label{clm:01}
For $i\in \{1, 2\}$, $2 M^+_i$ is a strongly base orderable matroid. 
\end{claim}
\begin{proof}
Fix $i \in \{1, 2\}$. We can easily see that $M^+_i$ is strongly base orderable.  Suppose that we are given two bases $X_1, X_2 \in 2 \mathcal{B}^+_i$, and  suppose also that $X_1 = Y^1_1 \cup Y^2_1$ and $X_2 = Y^1_2 \cup Y^2_2$, where $Y^1_1, Y^2_1, Y^1_2, Y^2_2 \in \mathcal{B}^+_i$. Since $M^+_i$ is strongly base orderable, for $j \in \{1, 2\}$, there exists a bijection $f_j: Y^j_1 \to Y^j_2$ such that $(Y^j_1 \setminus X) \cup f_j(X) \in \mathcal{B}^+_i$ for any $X \subseteq Y^j_1$. Then, $f_1$ and $f_2$ naturally define a bijection $f: X_1 \to X_2$ such that $(X_1 \setminus X) \cup f(X) \in 2\mathcal{B}^+_i$ for any $X \subseteq X_1$. This shows that $2 M^+_i$ is strongly base orderable. 
\end{proof}

Let $X_0 := (B_1 \cup B_2) \cup (B_1 \cap B_2)'$. Then, $X_0$ is a common basis of $2M^+_1$ and $2M^+_2$. We consider a bipartite digraph $D^+ = (V, E^+)$ defined by 
\begin{align*}
V &= (B_1 \cap B_2) \cup (S \setminus (B_1 \cup B_2)), \\
E^+ &= \{(x, y) \mid x \in B_1 \cap B_2,\, y \in S \setminus (B_1 \cup B_2),\, X_0 - x + y \in 2\mathcal{B}^+_1 \} \\
   & \qquad \cup \{(y, x) \mid x \in B_1 \cap B_2,\, y \in S \setminus (B_1 \cup B_2),\, X_0 - x + y \in 2\mathcal{B}^+_2 \}. 
\end{align*}

\begin{claim}\label{clm:02}
The digraph $D^+$ is acyclic. 
\end{claim}
\begin{proof}
Suppose to the contrary that $D^+$ contains a dicycle. 
Choose a dicycle $C$ with the smallest number of vertices, which implies that $C$ has no chord. 
Then, $X_0 \triangle V(C)$ is a common basis of $2M^+_1$ and $2M^+_2$ by Theorem~\ref{thm:krog}.  
By Theorem~\ref{thm:DM76} and Claim~\ref{clm:01}, 
$X_0 \triangle V(C)$ can be partitioned into two common bases of $M^+_1$ and $M^+_2$. 
Let $\tilde{B}_1$ and $\tilde{B}_2$ be the sets in $S$ corresponding to these common bases. 
Then, $\tilde{B}_1, \tilde{B}_2 \in \mathcal{B}_1 \cap \mathcal{B}_2$ and $|\tilde{B}_1 \cap \tilde{B}_2| < |B_1 \cap B_2|$. 
This contradicts that $|B_1 \cap B_2|$ is minimized.
\end{proof}

We now consider the digraph $D=(V,E)$ defined by (\ref{eq:D}). 
For $x\in B_1 \cap B_2$ and $y\in S \setminus (B_1 \cup B_2)$, we observe that $B_1-x+y\in \mathcal{B}_1$ implies $X_0-x+y\in 2\mathcal{B}^+_1$ and $B_2-x+y\in \mathcal{B}_2$ implies $X_0-x+y\in 2\mathcal{B}^+_2$. This shows that $D$ is a subgraph of $D^+$, and hence $D$ is acyclic by Claim~\ref{clm:02}. 
Therefore, we can find a function $p: S \to \mathbb{R}$ such that
$p(x) := 0$ for $x \in B_1 \setminus B_2$, 
$p(x) := |S|+1$ for $x \in B_2 \setminus B_1$, 
$p(x)$ are distinct values in $\{1, 2, \dots , |S|\}$ for $x \in V$, and 
$p(x) < p(y)$ for $(x, y) \in E$.
Then, Claim~\ref{clm:03b} shows that $\argmin_{X \in \mathcal{B}_1} p(X) = \{B_1\}$ and $\argmax_{X \in \mathcal{B}_2} p(X) = \{B_2\}$. 
Since $B_1, B_2 \in \mathcal{B}_1 \cap \mathcal{B}_2$, $p$ satisfies the requirements in Conjecture~\ref{conj:02}. Thus, Conjecture~\ref{conj:02} holds.

This proof can be converted to a polynomial-time algorithm for computing $p$ as follows. We first pick up two arbitrary common bases $B_1, B_2 \in \mathcal{B}_1 \cap \mathcal{B}_2$ and 
construct a digraph $D^+$ as above. 
If $D^+$ is acyclic, then we can find an appropriate function $p$. Otherwise, the proof of Claim~\ref{clm:02} shows that we can find $\tilde{B}_1, \tilde{B}_2 \in \mathcal{B}_1 \cap \mathcal{B}_2$ with $|\tilde{B}_1 \cap \tilde{B}_2| < |B_1 \cap B_2|$. Then, we update $B_i \leftarrow \tilde{B}_i$ for $i \in \{1,2\}$, construct $D^+$, and repeat this procedure. 
Since $|B_1 \cap B_2|$ decreases monotonically, this procedure is executed at most $|S|$ times.

Recall that Conjectures~\ref{conj:01} and \ref{conj:02} are equivalent by replacing $M_2$ with $M^*_2$. Since $M_2$ is strongly base orderable if and only if $M^*_2$ is strongly base orderable, Conjecture~\ref{conj:01} also holds for strongly base orderable matroids. 

This together with Lemma~\ref{lem:00to01} shows that Conjecture~\ref{conj:00} also holds. 
We note that, if $M_1$ and $M_2$ are strongly base orderbale matroids, then the matroids $M'_1$ and $M'_2$ obtained by deletion and contraction in the proof of Lemma~\ref{lem:00to01} are also strongly base orderable. 
\end{proof}

Finally in this section, we show an application of Theorem~\ref{thm:sbo} to bipartite matching, which is of independent interest. For a vertex $v$ in a graph, let $\delta(v)$ denote the set of all the edges incident to $v$.  

\begin{corollary}
For a bipartite graph $G=(U,V;E)$ containing a perfect matching, there exists a weight function $w: E \to \mathbb{R}$ satisfying the following conditions.
\begin{enumerate}
    \item For each $u \in U$, let $e_u$ be a lightest edge in $\delta(u)$ with respect to $w$. Then, $\{e_u \mid u \in U\}$ is a perfect matching in $G$. 
    \item For each $v \in V$, let $e_v$ be a heaviest edge in $\delta(v)$ with respect to $w$. Then, $\{e_v \mid v \in V\}$ is a perfect matching in $G$. 
\end{enumerate}
\end{corollary}

\begin{proof}
Let $\mathcal{B}_1=\{F \subseteq E \mid |F \cap \delta(u)|=1 \text{ for any } u \in U\}$ and $\mathcal{B}_2=\{F \subseteq E \mid |F \cap \delta(v)|=1 \text{ for any } v \in V\}$. By definition, $(E, \mathcal{B}_1)$ and $(E, \mathcal{B}_2)$ are partition matroids, and hence they are strongly base orderable matroids. Since Conjecture~\ref{conj:02} holds for strongly base orderable matroids and $\mathcal{B}_1 \cap \mathcal{B}_2$ is the set of perfect matchings in $G$, we obtain the corollary. 
\end{proof}

\section{Reduction from the weighted case to the unweighted case}\label{sec:weighted}

In this section, we show that the weighted problem can be reduced to the unweighted one, and prove Theorem~\ref{thm:unweight2weight}. 

\unweighttoweight*
\begin{proof}
Since Conjectures~\ref{conj:01} and~\ref{conj:02} are equivalent, it suffices to show that Conjecture~\ref{conj:03} is true by assuming that Conjecture~\ref{conj:02} is true. 

Suppose that we are given $M_i=(S, \mathcal{B}_i)$ and $w_i: S \to \mathbb{R}$ for $i \in \{1, 2\}$ as in Conjecture~\ref{conj:03}. We first consider the problem of finding a maximum weight common basis of $M_1$ and $M^*_2$ with respect to $w_1 - w_2$, where $M^*_2=(S, \mathcal{B}^*_2)$ is the dual matroid of $M_2$. By Theorem~\ref{thm:Frank}, there exist two functions $q_1: S \to \mathbb{R}$ and $q_2: S \to \mathbb{R}$ with $q_1 + q_2 = w_1 - w_2$ such that 
\begin{equation}
\argmax_{X \in \mathcal{B}_1 \cap \mathcal{B}^*_2} (w_1(X) - w_2(X)) = 
  \bigg( \argmax_{X \in \mathcal{B}_1} q_1(X) \bigg) \cap \bigg( \argmax_{X \in \mathcal{B}^*_2} q_2(X) \bigg). \label{eq:01}
\end{equation}

Define $\hat{\mathcal{B}_1} = \argmax_{X \in \mathcal{B}_1} q_1(X)$ and $\hat{\mathcal{B}_2} = \argmax_{X \in \mathcal{B}^*_2} q_2(X)$. 
Then, it is known that $\hat{M_i} = (S, \hat{\mathcal{B}_i})$ is also a matroid for $i\in\{1,2\}$ (see~\cite{edmonds1971matroids}). By (\ref{eq:01}), we obtain 
\begin{equation}
\argmax_{X \in \mathcal{B}_1 \cap \mathcal{B}^*_2} (w_1(X) - w_2(X)) = \hat{\mathcal{B}_1} \cap \hat{\mathcal{B}_2}. \label{eq:02}
\end{equation}
This together with $\mathcal{B}_1 \cap \mathcal{B}^*_2\not=\emptyset$
shows that $\hat{\mathcal{B}_1} \cap \hat{\mathcal{B}_2} \not= \emptyset$, and hence $\hat{M_1}$ and $\hat{M_2}$ satisfy the assumptions in Conjecture~\ref{conj:02}. Therefore, by assuming that Conjecture~\ref{conj:02} is true, there exists a function $\hat{p}: S \to \mathbb{R}$ satisfying the following conditions.  
\begin{enumerate}[(a)]
\item For any $B_1 \in \argmin_{X \in \hat{\mathcal{B}}_1} \hat{p}(X)$, it holds that $B_1 \in \hat{\mathcal{B}}_2$. \label{it:a} 
\item For any $B_2 \in \argmax_{X \in \hat{\mathcal{B}}_2} \hat{p}(X)$, it holds that $B_2 \in \hat{\mathcal{B}}_1$. \label{it:b}
\end{enumerate}

Let $\delta := \min \{|q_i(X) - q_i(Y)| \mid i \in \{1, 2\},\, X, Y \subseteq S,\, q_i(X) \not= q_i(Y)\}$ and let $\varepsilon$ be a positive number such that $\varepsilon \cdot |\hat{p}(X)| < \delta / 2$ for any $X \subseteq S$. We now show that $p :=  w_1 - q_1 + \varepsilon \cdot \hat{p}$ satisfies the requirements of Conjecture~\ref{conj:03}. Let $B_1$ be a set in
$ \argmax_{X \in \mathcal{B}_1} (w_1(X) - p(X)) =  \argmax_{X \in \mathcal{B}_1} (q_1(X) - \varepsilon \cdot \hat{p}(X))$. 
Since $- \delta /2 < \varepsilon \cdot \hat{p}(X) < \delta /2$ for any $X \subseteq S$, we have that  $B_1 \in \argmax_{X \in \mathcal{B}_1} q_1(X) = \hat{\mathcal{B}_1}$ and  $B_1 \in \argmin_{X \in \hat{\mathcal{B}_1}} \hat{p}(X)$. Then \eqref{it:a} shows that $B_1 \in \hat{\mathcal{B}_2}$. Therefore, 
$$
B_1 \in \hat{\mathcal{B}_1} \cap \hat{\mathcal{B}_2} = \argmax_{X \in \mathcal{B}_1 \cap \mathcal{B}^*_2} (w_1(X) - w_2(X)) = \argmax_{X \in \mathcal{B}_1 \cap \mathcal{B}^*_2} (w_1(X) + w_2(S \setminus X))
$$
holds by (\ref{eq:02}),  which means that $p$ satisfies the first requirement in Conjecture~\ref{conj:03}. 

Similarly, let $B_2$ be a set in 
$$
\argmax_{X \in \mathcal{B}_2} (w_2(X) - p(X)) =  \argmax_{X \in \mathcal{B}_2} (-q_2(X) - \varepsilon \cdot \hat{p}(X)) = \argmax_{X \in \mathcal{B}_2} (q_2(S \setminus X) + \varepsilon \cdot \hat{p}(S \setminus X)).
$$
This shows that  $S \setminus B_2 \in \argmax_{X \in \mathcal{B}^*_2} q_2(X) = \hat{\mathcal{B}_2}$ and  $S \setminus B_2 \in \argmax_{X \in \hat{\mathcal{B}_2}} \hat{p}(X)$. Then \eqref{it:b} shows that $S \setminus B_2 \in \hat{\mathcal{B}_1}$. Therefore, 
$$
S \setminus B_2 \in \hat{\mathcal{B}_1} \cap \hat{\mathcal{B}_2} = \argmax_{X \in \mathcal{B}_1 \cap \mathcal{B}^*_2} (w_1(X) - w_2(X)) = \argmax_{X \in \mathcal{B}_1 \cap \mathcal{B}^*_2} (w_1(X) + w_2(S \setminus X))
$$ 
holds by (\ref{eq:02}),  which means that $p$ satisfies the second requirement in Conjecture~\ref{conj:03}. 
Therefore, Conjecture~\ref{conj:03} is true if Conjecture~\ref{conj:02} is true. 
\end{proof}

\begin{remark}
Algorithmically, if we can compute $\hat p$, then we can compute $p$ efficiently as follows. 
First, we may assume that $w_1$ and $w_2$ are integral by multiplying by the common denominator. Then, we can take $q_1$ and $q_2$ so that they are integral~\cite{frank1981weighted}.  Therefore, we have that $\delta \ge 1$, and hence $\varepsilon := 1 / (1 + 2 \sum_{s \in S} |\hat p(s)|)$ satisfies the conditions in the proof. This shows that we can compute $p :=  w_1 - q_1 + \varepsilon \cdot \hat{p}$. 
\end{remark}

By Theorem~\ref{thm:unweight2weight}, we obtain Corollaries~\ref{cor:partitionw} and~\ref{cor:sbow} as follows. In the proof of Theorem~\ref{thm:unweight2weight}, we consider Conjecture~\ref{conj:02} for matroids $\hat{M_i} = (S, \hat{\mathcal{B}_i})$,  where $\hat{\mathcal{B}_1} = \argmax_{X \in \mathcal{B}_1} q_1(X)$ and $\hat{\mathcal{B}_2} = \argmax_{X \in \mathcal{B}^*_2} q_2(X)$. 
Observe that if $M_1$ is a partition matroid with partition classes of size at most $2$ and with all-ones upper bound on the partition classes, then so is $\hat{M_1}$. Furthermore, Lemma~\ref{lem:maximizerSBO} shows that if $M_i$ is strongly base orderable, then so is $\hat{M_i}$. Since Theorems~\ref{thm:partition} and~\ref{thm:sbo} imply that Conjecture~\ref{conj:02} also holds for these cases, we obtain Corollaries~\ref{cor:partitionw} and~\ref{cor:sbow}.



\section{Gross substitutes valuations}
\label{sec:Mnatural}

In this section, we show that the reduction technique in Section~\ref{sec:weighted} works also for M${}^\natural$-concave functions, or equivalently, gross substitutes functions.  M${}^\natural$-concave functions are introduced by Murota and Shioura~\cite{MurotaShioura1999} and play a central role in the theory of discrete convex analysis. A function $f:\mathbb{Z}^S \to \mathbb{R} \cup \{-\infty\}$ is said to be {\em M${}^\natural$-concave} if it satisfies the following exchange property: 
\begin{description}
\item[(M${}^\natural$-EXC)] 
$\forall x,y \in {\rm dom} f,\ \forall i \in \suppp(x-y),\ \exists j \in \suppm(x-y) \cup \{0\}$: 
$$
f(x)+f(y) \le f(x-\chi_i+\chi_j) + f(y+\chi_i-\chi_j),  
$$
\end{description}
where ${\rm dom} f=\{x \in \mathbb{Z}^S | f(x)>-\infty \}$, $\suppp(x)=\{i \in S \mid x(i)>0 \}$, $\suppm(x)=\{i \in S \mid x(i)<0 \}$ for $x \in \mathbb{Z}^S$, $\chi_i$ is the characteristic vector of $i \in S$, and $\chi_0$ is the all-zero vector ${\bf 0}$. When we consider a function $f$ on $\{0,1\}^S$, $f$ can be regarded as a function on $\mathbb{Z}^S$ by setting $f(x)=-\infty$ for $x \in \mathbb{Z}^S \setminus \{0,1\}^S$. It is shown by Fujishige and Yang~\cite{FujishigeYang2003} that a function $f$ on $\{0,1\}^S$ is M${}^\natural$-concave if and only if it is a gross substitutes function (see also \cite[Theorem 6.34]{murotaDCA}). See survey papers \cite{Murota2016,ShiouraTamura2015} for more details on M${}^\natural$-concave functions and gross substitutes functions. For a set $Q \subseteq \mathbb{Z}^S$, we define a function $f_Q$ on $\mathbb{Z}^S$ by $f_Q(x)=0$ if $x \in Q$ and $f_Q(x)=-\infty$ otherwise. We say that a set $Q \subseteq \mathbb{Z}^S$ is {\em M${}^\natural$-convex} if $f_Q$ is an M${}^\natural$-concave function. It is known that a set is M${}^\natural$-convex if and only if it is the set of integer points/vectors in an integral g-polymatroid~\cite{frank1984,frank2011}. 
Let ${\bf 1}$ denote the all-one vector in $\mathbb{Z}^S$.

We are interested in the existence of a pricing scheme for the two-buyer case with gross substitutes valuations (or equivalently, M${}^\natural$-concave valuations), which is stated as follows. 

\begin{conjecture} \label{conj:M01}
For $i=1, 2$, let $v_i: \{0, 1\}^S \to \bR \cup \{-\infty\}$ be an M${}^\natural$-concave function. Then, there exists a vector $p \in \mathbb{R}^S$ satisfying the following conditions.  
\begin{enumerate}
\item For any $x_1 \in \argmax_{x \in \{0,1\}^S} (v_1(x) - p \cdot x)$, it holds that $x_1 \in \argmax_{x \in \{0,1\}^S} (v_1(x) + v_2({\bf 1} - x))$. 
\item For any $x_2 \in \argmax_{x \in \{0,1\}^S} (v_2(x) - p \cdot x)$, it holds that $x_2 \in \argmax_{x \in \{0,1\}^S} (v_1({\bf 1} - x) + v_2(x))$. 
\end{enumerate}
\end{conjecture}

In Conjecture~\ref{conj:M01}, a subset of $S$ is represented by its characteristic vector. If buyer $i$ comes to a shop first, then she chooses an arbitrary set $x_i$ maximizing her utility $v_i(x) - p \cdot x$. Then, the second buyer takes the set of all the remaining items whose characteristic vector is ${\bf 1} - x_i$. Conjecture~\ref{conj:M01} asserts that, regardless of the choice of $x_i$, this mechanism gives an allocation maximizing the social welfare. 

As an unweighted version of this conjecture, we consider the following conjecture. 

\begin{conjecture} \label{conj:M02}
For $i=1, 2$, let $Q_i \subseteq \{0, 1\}^S$ be an M${}^\natural$-convex set such that there exist $x_1 \in Q_1$ and $x_2 \in Q_2$ with $x_1 + x_2 = {\bf 1}$. Then, there exists a vector $p \in \mathbb{R}^S$ satisfying the following conditions.  
\begin{enumerate}
\item For any $x_1 \in \argmin_{x \in Q_1} (p \cdot x)$, it holds that ${\bf 1} - x_1  \in Q_2$. 
\item For any $x_2 \in \argmin_{x \in Q_2} (p \cdot x)$, it holds that ${\bf 1} - x_2  \in Q_1$. 
\end{enumerate}
\end{conjecture}

In Conjecture~\ref{conj:M02}, each buyer $i$ has an admissible set $Q_i$ instead of a valuation. More precisely, each buyer $i$ wants to buy a set of items whose characteristic vector $x_i$ belongs to a given M${}^\natural$-convex set $Q_i$. We can easily see that Conjecture~\ref{conj:M02} is a special case of Conjecture~\ref{conj:M01}, in which $v_i=f_{Q_i}$ for $i=1, 2$. We now prove that the reverse implication also holds, which means that Conjecture~\ref{conj:M01}  can be reduced to the unweighted case.

\begin{restatable}{theorem}{MnaturalReduction}
\label{thm:MnaturalReduction}
If Conjecture~\ref{conj:M02} is true, then Conjecture~\ref{conj:M01} is also true.
\end{restatable}

\begin{proof}
Let $v^*_2: \{0, 1\}^S \to \bR \cup \{-\infty\}$ be the function defined by $v^*_2(x) = v_2({\bf 1} - x)$ for $x \in \{0, 1\}^S$. Then, $v^*_2$ is also an M${}^\natural$-concave function. Consider the problem of maximizing $v_1(x) + v^*_2(x)$ subject to $x \in \{0, 1\}^S$. By the M-convex intersection theorem (see~\cite[Theorem 8.17]{murotaDCA}), there exists a vector $q \in \bR^S$ such that 
\begin{equation}
\argmax_{x \in \{0,1\}^S} (v_1(x) + v^*_2(x)) =  \bigg( \argmax_{x \in \{0,1\}^S} (v_1(x) - q \cdot x) \bigg) \cap \bigg( \argmax_{x \in \{0,1\}^S} (v^*_2(x) + q \cdot x) \bigg). \label{eq:M01}
\end{equation}
Define $Q_1=\argmax_{x \in \{0,1\}^S} (v_1(x) - q \cdot x)$, $Q^*_2=\argmax_{x \in \{0,1\}^S} (v^*_2(x) + q \cdot x)$, and $Q_2=\{{\bf 1}-x \mid x \in Q^*_2\}$. Then, it is known that $Q_1$ and $Q^*_2$ are M${}^\natural$-convex sets (see~\cite[Theorem 6.30(2)]{murotaDCA}), and so is $Q_2$ (see~\cite[Theorem 6.13(2)]{murotaDCA}).  By (\ref{eq:M01}), we obtain $\argmax_{x \in \{0,1\}^S} (v_1(x) + v^*_2(x)) = Q_1 \cap Q^*_2$. This shows that $Q_1 \cap Q^*_2 \not= \emptyset$, and hence $Q_1$ and $Q_2$ satisfy the assumptions in Conjecture~\ref{conj:M02}. Therefore, by assuming that Conjecture~\ref{conj:M02} is true, there exists a vector $\hat{p} \in \mathbb{R}^S$ satisfying the following conditions.  
\begin{enumerate}
\item[(a)] For any $x_1 \in \argmin_{x \in Q_1} (\hat p \cdot x)$, it holds that ${\bf 1} - x_1  \in Q_2$. 
\item[(b)] For any $x_2 \in \argmin_{x \in Q_2} (\hat p \cdot x)$, it holds that ${\bf 1} - x_2  \in Q_1$. 
\end{enumerate}
Then, by the same argument as the proof of Theorem~\ref{thm:unweight2weight}, $p:= q + \varepsilon \cdot \hat p$ satisfies the requirements in Conjecture~\ref{conj:M01}, where $\varepsilon$ is a sufficiently small positive number. 
\end{proof}

\begin{remark}
In a market model, it is common to assume that each valuation $v_i$ is monotone and $v_i(\emptyset)=0$. We note that these assumptions are not required in the proof of Theorem~\ref{thm:MnaturalReduction}. In return for this, the obtained price vector $p$ is not necessarily non-negative. 
\end{remark}

We note that Conjecture~\ref{conj:01} is a special case of Conjecture~\ref{conj:M01}, as the characteristic vector of all the bases of a matroid forms an M${}^\natural$-convex set. This relationship supports the importance of Conjecture~\ref{conj:01}.

\section{Conclusion}\label{sec:conc}

We considered the existence of prices that are capable to achieve optimal social welfare without a central tie-breaking coordinator. Although such pricing looks similar to well-known Walrasian pricing, it is less understood even for two-buyer markets with gross substitute valuations.
This paper focuses on two-buyer markets with rank valuations, and we gave polynomial-time algorithms that always find such prices when one of the matroids is a simple partition matroid or both matroids are strongly base orderable. This result partially answers a question of D\"uetting and V\'egh. We further formalized a weighted variant of the conjecture of D\"uetting and V\'egh, and showed that the weighted variant can be reduced to the unweighted one based on the weight-splitting theorem of Frank. We also showed that a similar reduction technique works for M${}^\natural$-concave functions, or equivalently, gross substitutes functions.


\paragraph{Acknowledgement}

Krist\'of B\'erczi was supported by the J\'anos Bolyai Research Fellowship of the Hungarian Academy of Sciences and by Projects no.~NKFI-128673 and no.~ED\_18-1-2019-0030 provided from the National Research, Development and Innovation Fund of Hungary. Naonori Kakimura was supported by JSPS KAKENHI Grant Numbers JP17K00028 and JP18H05291, Japan. Yusuke Kobayashi was supported by JSPS KAKENHI grant numbers JP18H05291, JP19H05485, and JP20K11692, Japan. The work was supported by the Research Institute for Mathematical Sciences, an International Joint Usage/Research Center located in Kyoto University.

\bibliographystyle{abbrv}
\bibliography{pricing}

\begin{thebibliography}{10}

\bibitem{berger2020power}
B.~Berger, A.~Eden, and M.~Feldman.
\newblock On the power and limits of dynamic pricing in combinatorial markets.
\newblock {\em arXiv preprint arXiv:2002.06863}, 2020.

\bibitem{blumrosen2008posted}
L.~Blumrosen and T.~Holenstein.
\newblock Posted prices vs. negotiations: an asymptotic analysis.
\newblock {\em EC}, 10:1386790--1386801, 2008.

\bibitem{chawla2010multi}
S.~Chawla, J.~D. Hartline, D.~L. Malec, and B.~Sivan.
\newblock Multi-parameter mechanism design and sequential posted pricing.
\newblock In {\em Proceedings of the forty-second ACM symposium on Theory of
  computing}, pages 311--320, 2010.

\bibitem{chawla2010power}
S.~Chawla, D.~L. Malec, and B.~Sivan.
\newblock The power of randomness in bayesian optimal mechanism design.
\newblock In {\em Proceedings of the 11th ACM conference on Electronic
  commerce}, pages 149--158, 2010.

\bibitem{chawla2019pricing}
S.~Chawla, J.~B. Miller, and Y.~Teng.
\newblock Pricing for online resource allocation: intervals and paths.
\newblock In {\em Proceedings of the Thirtieth Annual ACM-SIAM Symposium on
  Discrete Algorithms}, pages 1962--1981. SIAM, 2019.

\bibitem{clarke1971multipart}
E.~H. Clarke.
\newblock Multipart pricing of public goods.
\newblock {\em Public choice}, pages 17--33, 1971.

\bibitem{cohen2016invisible}
V.~Cohen-Addad, A.~Eden, M.~Feldman, and A.~Fiat.
\newblock The invisible hand of dynamic market pricing.
\newblock In {\em Proceedings of the 2016 ACM Conference on Economics and
  Computation}, pages 383--400, 2016.

\bibitem{davies1976disjoint}
J.~Davies and C.~McDiarmid.
\newblock Disjoint common transversals and exchange structures.
\newblock {\em Journal of the London Mathematical Society}, 2(1):55--62, 1976.

\bibitem{dutting2016posted}
P.~D{\"u}tting, M.~Feldman, T.~Kesselheim, and B.~Lucier.
\newblock Posted prices, smoothness, and combinatorial prophet inequalities.
\newblock {\em arXiv preprint arXiv:1612.03161}, 2016.

\bibitem{dutting2017prophet}
P.~D\"utting, M.~Feldman, T.~Kesselheim, and B.~Lucier.
\newblock Prophet inequalities made easy: Stochastic optimization by pricing
  non-stochastic inputs.
\newblock In {\em 2017 IEEE 58th Annual Symposium on Foundations of Computer
  Science (FOCS)}, pages 540--551. IEEE, 2017.

\bibitem{vegh}
P.~D\"utting and L.~A. V\'egh.
\newblock {Private Communication}, 2017.

\bibitem{eden2019max}
A.~Eden, U.~Feige, and M.~Feldman.
\newblock Max-min greedy matching.
\newblock {\em Approximation, Randomization, and Combinatorial Optimization.
  Algorithms and Techniques (APPROX/RANDOM 2019)}, 145:7, 2019.

\bibitem{edmonds1971matroids}
J.~Edmonds.
\newblock Matroids and the greedy algorithm.
\newblock {\em Mathematical programming}, 1(1):127--136, 1971.

\bibitem{ezra2018pricing}
T.~Ezra, M.~Feldman, T.~Roughgarden, and W.~Suksompong.
\newblock Pricing multi-unit markets.
\newblock In {\em International Conference on Web and Internet Economics},
  pages 140--153. Springer, 2018.

\bibitem{feldman2014combinatorial}
M.~Feldman, N.~Gravin, and B.~Lucier.
\newblock Combinatorial auctions via posted prices.
\newblock In {\em Proceedings of the twenty-sixth annual ACM-SIAM symposium on
  Discrete algorithms}, pages 123--135. SIAM, 2014.

\bibitem{feldman2016combinatorial}
M.~Feldman, N.~Gravin, and B.~Lucier.
\newblock Combinatorial walrasian equilibrium.
\newblock {\em SIAM Journal on Computing}, 45(1):29--48, 2016.

\bibitem{frank1981weighted}
A.~Frank.
\newblock A weighted matroid intersection algorithm.
\newblock {\em Journal of Algorithms}, 2(4):328--336, 1981.

\bibitem{frank1984}
A.~Frank.
\newblock Generalized polymatroids.
\newblock In {\em Finite and Infinite Sets (Colloquia Mathematica Societatis
  J{\'a}nos Bolyai 37)}, pages 285--294, 1984.

\bibitem{frank2011}
A.~Frank.
\newblock {\em Connections in Combinatorial Optimization}.
\newblock Oxford University Press, 2011.

\bibitem{fujishige2003note}
S.~Fujishige and Z.~Yang.
\newblock A note on {Kelso and Crawford's} gross substitutes condition.
\newblock {\em Mathematics of Operations Research}, 28(3):463--469, 2003.

\bibitem{groves1973incentives}
T.~Groves.
\newblock Incentives in teams.
\newblock {\em Econometrica: Journal of the Econometric Society}, pages
  617--631, 1973.

\bibitem{gul1999walrasian}
F.~Gul, E.~Stacchetti, et~al.
\newblock Walrasian equilibrium with gross substitutes.
\newblock {\em Journal of Economic theory}, 87(1):95--124, 1999.

\bibitem{hsu2016prices}
J.~Hsu, J.~Morgenstern, R.~Rogers, A.~Roth, and R.~Vohra.
\newblock Do prices coordinate markets?
\newblock In {\em Proceedings of the forty-eighth annual ACM symposium on
  Theory of Computing}, pages 440--453, 2016.

\bibitem{huang1994relations}
R.~Huang and G.-C. Rota.
\newblock On the relations of various conjectures on latin squares and
  straightening coefficients.
\newblock {\em Discrete Mathematics}, 128(1-3):225--236, 1994.

\bibitem{kelso1982job}
A.~S. Kelso~Jr and V.~P. Crawford.
\newblock Job matching, coalition formation, and gross substitutes.
\newblock {\em Econometrica: Journal of the Econometric Society}, pages
  1483--1504, 1982.

\bibitem{krogdahl1974combinatorial}
S.~Krogdahl.
\newblock A combinatorial base for some optimal matroid intersection
  algorithms.
\newblock Technical Report STAN-CS-74-468, Computer Science Department,
  Stanford University, Stanford, CA, U.S., 1974.

\bibitem{krogdahl1976combinatorial}
S.~Krogdahl.
\newblock A combinatorial proof for a weighted matroid intersection algorithm.
\newblock Technical Report Computer Science Report 17, Institute of
  Mathematical and Physical Sciences, University of Tromso, Tromso, Norway,
  1976.

\bibitem{krogdahl1977dependence}
S.~Krogdahl.
\newblock The dependence graph for bases in matroids.
\newblock {\em Discrete Mathematics}, 19(1):47--59, 1977.

\bibitem{murotaDCA}
K.~Murota.
\newblock {\em Discrete Convex Analysis}.
\newblock Society for Industrial and Applied Mathematics, 2003.

\bibitem{Murota2016}
K.~Murota.
\newblock Discrete convex analysis: A tool for economics and game theory.
\newblock {\em The Journal of Mechanism and Institution Design}, 1(1):151--273,
  2016.

\bibitem{MurotaShioura1999}
K.~Murota and A.~Shioura.
\newblock M-convex function on generalized polymatroid.
\newblock {\em Mathematics of Operations Research}, 24(1):95--105, 1999.

\bibitem{nisan2006communication}
N.~Nisan and I.~Segal.
\newblock The communication requirements of efficient allocations and
  supporting prices.
\newblock {\em Journal of Economic Theory}, 129(1):192--224, 2006.

\bibitem{nishimura2009lost}
H.~Nishimura and S.~Kuroda.
\newblock {\em A Lost Mathematician, Takeo Nakasawa: The Forgotten Father of
  Matroid Theory}.
\newblock Springer Science \& Business Media, 2009.

\bibitem{oxley2011matroid}
J.~Oxley.
\newblock {\em Matroid Theory}.
\newblock Oxford University Press, 2011.

\bibitem{schrijver2003combinatorial}
A.~Schrijver.
\newblock {\em Combinatorial optimization: polyhedra and efficiency},
  volume~24.
\newblock Springer Science \& Business Media, 2003.

\bibitem{ShiouraTamura2015}
A.~Shioura and A.~Tamura.
\newblock Gross substitutes condition and discrete concavity for multi-unit
  valuations: a survey.
\newblock {\em Journal of the Operations Research Society of Japan},
  58(1):61--103, 2015.

\bibitem{vickrey1961counterspeculation}
W.~Vickrey.
\newblock Counterspeculation, auctions, and competitive sealed tenders.
\newblock {\em The Journal of finance}, 16(1):8--37, 1961.

\bibitem{walras1896elements}
L.~Walras.
\newblock {\em {\'E}l{\'e}ments d'{\'e}conomie politique pure, ou, Th{\'e}orie
  de la richesse sociale}.
\newblock F. Rouge, 1896.

\bibitem{whitney1992abstract}
H.~Whitney.
\newblock On the abstract properties of linear dependence.
\newblock In {\em Hassler Whitney Collected Papers}, pages 147--171. Springer,
  1992.

\end{thebibliography}

\end{document}